\documentclass[13pt]{amsart}
\usepackage{amsmath,amssymb,amsfonts,amsthm,longtable}
\usepackage{amscd,pstricks}
\usepackage{graphics}
\usepackage{graphicx}
\usepackage{pstricks,pstricks-add,pst-math,pst-xkey}
\usepackage{graphicx}

\newtheorem{thm}{Theorem }[section]

\newtheorem{lem}[thm]{Lemma }

\newtheorem{prop}[thm]{Proposition }

\theoremstyle{definition}
\newtheorem{defn}[thm]{Definition }

\def \bra#1\ket {\mathop{\vphantom{#1}\left<\smash{#1}\right>}\nolimits}

\renewcommand \phi {\varphi}
 \numberwithin{equation}{section}

\begin{document}
\author{Bruce Lionnel LIETAP NDI}
\address[LIETAP]{University of Maroua\\
Faculty of Sciences, Department of Mathematics and Computer Sciences
\\ P.O. Box 814 Maroua, Cameroon} \email{nbruce.lionnel@gmail.com}

\author{Djagwa DEHAINSALA}
\address[DEHAINSALA]{University of NDjamena\\
Faculty of Exact and Applied Sciences, Department of Mathematics\\
1 route de Farcha, PO Box 1027 NDjamena, Chad}
\email{djagwa73@gmail.com}

\author{Joseph DONGHO}
\address[DONGHO]{University of Maroua\\
Faculty of Sciences, Department of Mathematics and Computer Sciences
\\ P.O. Box 814 Maroua, Cameroon} \email{josephdongho@yahoo.fr}

\thanks{This work was completed with the support of the Pr Joseph Dongho and Dr Djagwa Dehainsala.}
\thanks{ Corresponding author: University of Maroua-Cameroon \& University of NDjamena-Chad.}

\title[Linearization, separability and Lax pairs representation of
\(a_4^{\left(2\right)}\) Toda lattice]{Linearization, separability
and Lax pairs representation of \(a_4^{\left(2\right)}\) Toda
lattice}


\date{\today}


\begin{abstract}
The aim of this work is focused on linearizing and found the Lax
Pairs of the algebraic complete integrability (a.c.i) Toda lattice
associated with the twisted affine Lie algebra
\(a_4^{\left(2\right)}\). Firstly, we recall that our case of a.c.i
is a two-dimensional algebraic completely integrable systems for
which the invariant (real) tori can be extended to complex algebraic
tori (abelian surfaces). This implies that the geometry can be used
to study this system. Secondly, we show that the lattice is related
to the Mumford system and we construct an explicit morphism between
these systems, leading to a new Poisson structure for the Mumford
system. Finally, we give a new Lax equation for this Toda lattice
and we construct an explicit linearization of the system.
\end{abstract}

\maketitle

\section{Introduction}

Many integrable systems from classical mechanics admit a
complexification, where phase space and time are complexified, and
the geometry of the (complex) momentum map is the best possible
complex analogue of the geometry that appears in the Liouville
Theorem. Namely, in many relevant examples the generic complexified
fiber is an affine part of an Abelian variety (a compact algebraic
torus) and the integrable vector fields are translation invariant,
when restricted to any of these tori. Such integrable systems are
call them algebraic completely integrable systems, following the
original definition of Adler and van Moerbeke.

Integrable systems have been integrated classically in terms of
quadratures, usually through a sequence of very ingenious algebraic
manipulations especially tailored to the problem. More recently, it
was realized that whenever a system could be represented as a family
of Lax pairs. the system could be linearized on the Jacobian of a
spectral curve, defined by the characteristic polynomial of one of
the matrices in the Lax pair.

To show that a Hamiltonian system linearizes on an Abelian variety,
one may either construct a Lax representation of the differential
equation depending on an extra-parameter and linearize on the
Jacobian of the curve specified by its characteristic equation, or
one may complete the complexified invariant manifolds by using the
Laurent solutions of the differential equations. The latter method
allows us in addition to identify the nature of the invariant
manifolds and of the solutions of the system: in most examples the
isospectral manifolds and the invariant manifolds are different.

In the previous work \cite{lietap}, we have prove that the
\(a_4^{\left(2\right)}\) is a two-dimensional integrable system.
This system satisfies the linearization criterion [\cite{al_Moer3},
theorem 6.41] and it is an algebraic completely integrable in the
Adler-van Moerbeke sense. This system has a smooth hyperelliptic
curve of genus two. According to Vanhaecke \cite{Van} and Mumford's
description of hyperelliptic Jacobians (see [\cite{mumford}, Section
3.1]), like $\Gamma$ is a hyperelliptic curve of genus two then the
Riemann surface $\overline{\Gamma}$ is embedded in its jacobian such
that $Jac(\overline{\Gamma})\ \Gamma$ is isomorphic to the space of
pairs of polynomials $(u(\lambda); v(\lambda))$. $u(\lambda)$ is a
monic of degree two and $v(\lambda)$ less than two. $f(\lambda)-
v^2(\lambda)$ is divisible by $u(\lambda)$.

The aim of this paper is how we can linearize and find the Toda
lattice $a_4^{(2)}$ Lax pair or Lax representation? To prove this,
we construct an explicit map from the generic fiber $\mathbb{F}_c$
into the Jacobian of the Riemann surface $\overline{\Gamma}_c$.
After we find the kummer surface of $Jac(\mathcal{K}_c)$,
$u(\lambda), v(\lambda)$ and $f(\lambda)$.

This paper is organized as follows. In section 2, preliminaries of
this work, we give the basic notions of linearising, separating
variables and Lax representation. In section 3, main part of the
paper, we show that the a.c.i \(a_4^{\left(2\right)}\) Toda lattice
 is related to the Mumford system and we construct an explicit
morphism between these systems, leading to a new Poisson structure
for the Mumford system. Finally, we give a new Lax equation with
spectral parameter for this Toda lattice and we construct an
explicit linearization of the system.

\section{Preliminaries}

Let \(\mathbb{C}^n\) denote a complex vector space of dimension
\(n\).
\begin{defn}\cite{herbert}
A lattice in \(\mathbb{C}^n\) is a discrete subgroup of maximal rank
in \(\mathbb{C}^n\). It is a free abelian group of rank \(2\).\\

 A lattice \(\Lambda\) in \(\mathbb{C}^n\) acts in a natural way on the vector space \(\mathbb{C}^n\) and the
quotient \(\mathbb{T}^n=\mathbb{C}^n/ \Lambda\) is called a complex
torus.
\end{defn}

In the theory of linear algebraic groups there is the notion of a
torus. Such a torus is an affine group, whereas a complex torus is
compact.

\begin{defn}\cite{herbert}
An abelian variety is a complex torus admitting a positive line
bundle or equivalently a projective embedding.
\end{defn}
Abelian varieties over the complex numbers are special complex tori,
that is, quotients of finite-dimensional complex vector spaces
modulo a lattice of maximal rank. \\

The Riemann Relations are necessary and sufficient conditions for a
complex torus to be an abelian variety. They were introduced by
Riemann in the special case of a Jacobian variety of a curve.

Let \(\mathbb{T}^n=\mathbb{C}^n/ \Lambda\) be a complex torus.

\begin{defn}\cite{herbert}
A positive line bundle on \(\mathbb{T}^n\) is by definition a line
bundle on \(\mathbb{T}^n\) whose first Chern class is a positive
definite hermitian form on \(\mathbb{C}^n\). \\

A polarization on \(\mathbb{T}^n\) is by definition the first Chern
class \(H = c_1(L)\) of a positive line bundle \(L\) on
\(\mathbb{T}^n\).
\end{defn}

By abuse of notation we sometimes consider the line bundle \(L\) on
\(\mathbb{T}^n\) itself as a polarization. The type of \(L\) is
called the type of the polarization. A polarization is called
principal if it is of type \((1, \cdot\cdot\cdot , 1)\).\\

\begin{defn}\cite{herbert}
 An abelian variety is a complex torus \(\mathbb{T}^n\) admitting a
polarization \(H = c_1(L)\). The pair \((\mathbb{T}^n, H)\) is
called a polarized abelian variety.
\end{defn}

According to \cite{herbert}, let \(\Gamma\) be a smooth projective
curve of genus \(g\) over the field of complex numbers. the
\(g\)-dimensional \(\mathbb{C}\)-vector space
\(H^0(\omega_{\Gamma})\) of holomorphic \(1\)-forms on \(\Gamma\).
The homology group \(H^1(\Gamma, \mathbb{Z})\) is a free abelian
group of rank \(2g\). For convenience we use the same letter for
(topological) \(1\)-cycles on \(\Gamma\) and their corresponding
classes in  \(H^1(\Gamma, \mathbb{Z})\). By Stoke's theorem any
element \(\gamma \in H^1(\Gamma, \mathbb{Z})\) yields in a canonical
way a linear form on the vector space \(H^0(\omega_{\Gamma})\),
which we also denote by:
\[\begin{array}{cccc}
    \gamma: & H^0(\omega_{\Gamma})& \longrightarrow & \mathbb{C} \\
            & \omega              & \longmapsto     & \int_{\gamma}\omega
  \end{array}\]

\begin{defn}\cite{herbert}
  the Jacobian variety or simply the Jacobian of \(\Gamma\), denote
  by \(Jac(\Gamma)\) is a complex torus of dimension \(g\) such that
  \[Jac(\Gamma):=H^0(\omega_{\Gamma})^{\ast}/H^1(\Gamma,
  \mathbb{Z})\]
\end{defn}

\begin{defn}\cite{herbert}
A theta divisor of the Jacobian \(Jac(\Gamma)\) is any divisor on
\(Jac(\Gamma)\) such that the line bundle
\(\mathcal{O}_{Jac(\Gamma)}(\Theta)\) defines the canonical
polarization.
\end{defn}

\begin{defn}\cite{Piovan}
A system of ordinary differential equations over $\mathbb{R}$ is
called algebraic complete integrable (a.c.i.) when it is completely
integrable and the complexified invariant manifolds complete into
algebraic tori (Abelian varieties), whose (complexified) commuting
flows extend holomorphically.
\end{defn}

According to \cite{Piovan}, Let \(\mathbb{T}^n=\mathbb{C}^n/
\Lambda\) be a complex algebraic torus, (Abelian variety) with an
origin $0$ chosen. Let $i$ be the inverse morphism which coincides
with the $(-1)$-reflection about $0$.

\begin{defn}\cite{Piovan}
The Kumrner variety of \(\mathbb{T}^n\), denoted by $\mathcal{K}_c$,
is the quotient of \(\mathbb{T}^n\) by the action of the group $(1,
i)$.
\end{defn}

The Kummer variety bears the moduli information and has the
advantage of possessing a lower degree of embedding in projective
space. According to \cite{Piovan}, let \(\mathcal{D}\) be a divisor
on \(\mathbb{T}^n\). Denote by \(\mathcal{L}(\mathcal{D})\) the
invertible sheaf associated to \(\mathcal{D}\).
\begin{eqnarray*}
  \mathcal{L}(\mathcal{D}) &=& \mbox{\{ the vector space
of functions \(f\) such that } \\
   & & (f)= \mbox{divisor of zeroes-divisor of poles} \geq -
   \mathcal{D}\}
\end{eqnarray*}
According to \cite{Van} Let $\Gamma$ be a smooth curve of genus $g$.
We define two divisor
 \(\mathcal{D}\) and \(\mathcal{D}'\) in \(Div(\Gamma)\), the
 divisor group of \(\Gamma\), to be \emph{linearly equivalent},
 \(\mathcal{D}\sim_l \mathcal{D}'\), if and only if there exists a
 meromorphic function $f$ on \(\Gamma\).

 According to \cite{Piovan}, let \(\mathcal{D}\) be an ample divisor on \(\mathbb{T}^n\).
 We denote by \(\mathcal{C(D)}\) the set of all divisors \(\mathcal{D}'\) on
\(\mathbb{T}^n\) such that there are two positive numbers \(n, n'\)
and \(n\mathcal{D}\) is algebraically equivalent to
\(n'\mathcal{D}'\).

\begin{defn}\cite{al_Moer3}
A compact Riemann surface for which the Kodaira map is not an
embedding is called a hyperelliptic Riemann surface (a compact
Riemann surfaces of genus \(1\) being called an elliptic Riemann
surface), while any curve whose (compact) Riemann surface is
hyperelliptic is called a hyperelliptic curve (one speaks of an
elliptic curve in the genus \(1\) case).
\end{defn}

\section{Separability and linearization of two-dimensional Toda lattice \(a_4^{(2)}\)}

\subsection{Linearization procedure }
According to \cite{al_Moer3}, since \(Jac(\Gamma)\) is a principally
polarized Abelian variety of dimension \(g\), the Lefschetz Theorem
implies that it can be embedded in \(\mathbb{P}^{3^g - 1}\) , by
using the sections of \(\left[3\Theta\right]\). However, the
sections of \(\left[2\Theta\right]\) never embed \(Jac(\Gamma)\) in
projective space, but rather they embed its Kummer variety
\(K_c(\Gamma)\) in projective space. An important particular case is
that of the Kummer surface \(K_c(\Gamma)\), where \(\Gamma\) is a
hyperelliptic Riemann surface of genus \(2\). The line bundle
\(\left[2\Theta\right]\) that corresponds to twice the principal
polarization on \(Jac(\Gamma)\) has in this case \(4\) independent
sections and the associated Kodaira map, which maps \(Jac(\Gamma)\)
into \(\mathbb{P}^3\) , factors through \(K_c(\Gamma)\), realizing
the Kummer surface as a surface in \(\mathbb{P}^3\).

Being two-dimensional the image is given by a single equation; to
compute the degree of this equation, we use the fact that this
degree is given by \(\displaystyle\int_{K_c(\Gamma)}\omega\), where
\(\omega\) is associated \((1,1)\)-form ofthe standard Kahler
structure on \(\mathbb{P}^3\) . Clearly this is twice the volume of
\(K_c(\Gamma)\), which itself is half the volume of the Jacobi
surface (with the polarization of type \((1, 1)\)).

In the two-dimensional case, the invariant manifolds complete into
Abelian surfaces by adding one (or several) curves to the affine
surfaces. In this case, Vanhaecke proposed in \cite{Van} a method
which leads to an explicit linearization of the vector field of the
a.c.i. system. The computation of the first few terms of the Laurent
solutions to the differential equations enables us to construct an
embedding of the invariant manifolds in the projective space
\(\mathbb{P}^N\). From this embedding, one deduces the structure of
the divisors \(\mathcal{D}_c\) to be adjoined to the generic affine
in order to complete them into Abelian surfaces \(\mathbb{T}_c\).
Thus, the system is a.c.i.. The different steps of the algorithm of
Vanhaecke are given by:\\
\textbf{\underline{case 1}}
 \begin{itemize}
        \item [a)]If one of the components of \(\mathcal{D}_c\) is a
smooth curve \(\Gamma_c\) of genus two, compute the image of the
rational map \(\phi_{[2\Gamma_c]} : \mathbb{T}_c^2\rightarrow
\mathbb{P}^3\) which is a singular surface in \(\mathbb{P}^3\), the
Kummer surface \(\mathcal{K}_c\) of jacobian \(Jac(\Gamma_c)\) of
the curve \(\Gamma_c\).
        \item [b)] Otherwise, if one of the components of \(\mathcal{D}_c\)
        is a $d : 1$ unramified cover $\mathcal{C}_c$ of a smooth
curve \(\Gamma_c\) of genus two, the map $p :
\mathcal{C}_c\rightarrow \Gamma_c$ extends to the map $\widetilde{p}
: \mathbb{T}_c^2\rightarrow Jac(\Gamma_c)$. In this case, let
$\mathcal{C}_c$ denote the (non complete) linear system
$\widetilde{p}[2\Gamma_c] \subset [2\mathcal{C}_c]$ which
corresponds to the complete linear system $[2\mathcal{C}_c]$ and
compute now the Kummer surface $\mathcal{C}_c$ of $Jac(\Gamma_c)$ as
image of $\phi_{\varepsilon_c}: \mathbb{T}_c^2\rightarrow
\mathbb{P}^3$.
        \item [c)] Otherwise, change the divisor at infinity so as to arrive
        in case (a) or (b). This can always
be done for any irreducible Abelian surface.
    \end{itemize}
 \textbf{\underline{case 2}}. Choose a Weierstrass point $W$ on the curve $\Gamma_c$ and
    coordinates $(z_0 : z_1 : z_2 : z_3)$ for $\mathbb{P}^3$ such
$\phi_{[2\Gamma_c]}(W) = (0 : 0 : 0 : 1)$ in case 1.(a) and
$\phi_{\varepsilon_c}(W) = (0 : 0 : 0 : 1)$ in case 1.(b). Then this
point will be a singular point (node) for the Kummer surface
$\mathcal{K}_c$ whose equation is $\small{p_2(z_o; z_1; z_2)z^2_3 +
p_3(z_o; z_1; z_2)z_3 + p_4(z_o; z_1; z_2) = 0}$\\

 where the $p_i$
are polynomials of degree $i$. After a projective transformation
which fixes $(0 : 0 : 0 : 1)$, we may assume that $p_2(z_o; z_1;
z_2) = z^2_1-4z_0z_2$.\\

\textbf{\underline{case 3}}. Finally, let $s_1$ and $s_2$ be the
roots of the quadractic
    equation $z_0s^2 + z_1s + z_2 = 0$, whose
discriminant is $p^2(z_o; z_1; z_2)$, with the $z_i$ expressed in
terms of the original variables. Then the differential equations
describing the vector field of the system are rewritten by direct
computation in the classical Weierstrass form
\begin{equation}\label{systmumford}
\begin{array}{l}
  \frac{\dot{s_1}}{\sqrt{f(s_1)}}+\frac{\dot{s_2}}{\sqrt{f(s_2)}}= \alpha_1dt\\
  \frac{s_1\dot{s_1}}{\sqrt{f(s_1)}}+\frac{s_2\dot{s_2}}{\sqrt{f(s_2)}}= \alpha_2dt
\end{array}
\end{equation}
where $\alpha_1$ and $\alpha_2$ depend on the torus. From it, the
symmetric functions $s_1+s_2:= -\frac{z_1}{z_0}$ , $s_1s_2:=
\frac{z_2}{z_0}$ and the original variables can be written in terms
of the Riemann theta function associated to the curve $y^2 = f(x)$.

\subsection{A.C.I of \(a_4^{\left(2\right)}\) Toda lattice}
In this section, we recall, according to \cite{lietap}, some results
relating the two-dimensional \(a_4^{\left(2\right)}\) Toda
lattice. It is well known that this system is a.c.i. \\

The Toda lattice, introduced by Morikazu Toda in \(1967\)
\cite{toda}, is a simple model for a one-dimensional crystal in
solid-state physics. It is famous because it is one of the first
examples of a completely integrable nonlinear system. It is
described by a chain of particles with nearest-neighbor interaction,
and its dynamics are governed by the Hamiltonian \[
H\left(p,q\right)=\displaystyle \sum_{n\in \mathbb{Z}}\left(\frac{
p^2\left(n,t\right)}{2}+V\left(q\left(n+1,t\right)-q\left(n,t\right)\right)\right),
\] and the equations of motion
\[\left\{
  \begin{array}{ll}
    \frac{d}{dt}p\left(n,t\right)=-\frac{\partial
H\left(p,q\right)}{\partial q\left(n,t\right)}=
e^{-\left(q\left(n,t\right)-q\left(n-1,t\right)\right)}-e^{-\left(q\left(n+1,t\right)-q\left(n,t\right)\right)}
 \\
    \frac{d}{dt}q\left(n,t\right)=\frac{\partial
H\left(p,q\right)}{\partial p\left(n,t\right)}=p\left(n,t\right)
  \end{array}
\right.\]
 where \(q \left(n, t\right)\) is the displacement of the \(n\)-th particle from its equilibrium position,
 and \(p \left(n, t\right)\) is its momentum (with mass \(m = 1\)), and the Toda potential is given
 by \(V\left(r\right)=e^{-r}+r-1\). The classical Toda lattice is a system of particles with unit mass,
 connected by exponential springs. Its equations of motion derived from the Hamiltonian.
\begin{equation}
   H=\frac{1}{2}\sum_{j=1}^n p_j^2+ \sum_{j=1}^{n-1}
   e^{q_j-q_{j+1}}. \label{hamilttodaclasnperio}
\end{equation}
where \(q_j\) is the position of the \emph{j}-th particle and
\(p_j\) is its amount of movement. This type of Hamiltonian was
considered first by Morikazu Toda \cite{toda}. The equation
\eqref{hamilttodaclasnperio} is known as the finite classic no
periodic Toda lattice to distinguish other versions of various forms
of the system. The periodic version of \eqref{hamilttodaclasnperio}
is given by

\[H=\frac{1}{2}\sum_{j=1}^n p_j^2+ \sum_{j=1}^n e^{q_j-q_{j+1}}, q_{n+1}
    =q_1.
\]
 where the equations of motion are given by
\[
\dot{p}_j=-\frac{\partial H}{\partial q_j}=
e^{\left(q_{j-1}-q_j\right)}-e^{\left(q_j-q_{j+1}\right)} \mbox{ and
} \dot{q}_j=\frac{\partial H }{\partial p_j}=p_j,1\leq j\leq n .
\]
 The differential equations of the periodic Toda lattice
\(a_4^{\left(2\right)}\) are given on the five dimensions hyperplane
 \(\mathcal{H} =\{ \left( x_0 , x_1 , x_2 , y_0 , y_1 , y_2 \right)\in
\mathbb{C}^6 | y_0 + 2 y_1 + 2y_2 = 0 \} \mbox{ of } \mathbb{C}^6\)
by
\[
     \left\{
       \begin{array}{l}
         \dot{x}=x.y\\
         \dot{y}=Ax
       \end{array}
     \right.\]

where \(x=\left(x_0,x_1,x_2\right)^{\top}\) ,
\(y=\left(y_0,y_1,y_2\right)^{\top}\) and \(A\) is the Cartan matrix
of the twisted affine Lie algebra \(a_4^{\left(2\right)}\) given in
\cite{al_Moer3} by

\[
    \left(
   \begin{array}{ccc}
     2 & -2 & 0 \\
     -1 & 2 & -2 \\
     0 & -1 & 2 \\
   \end{array}
\right)\] and \(\varepsilon=\left(1,2,2\right)^{\top}\) is the
normalized null vector of \(A^{\top}\).
  The equations of motion of the Toda lattice \(a_4^{\left(2\right)}\) are given
 in \cite{al_Moer3}  by  :

\begin{equation}\label{systa4}
     \begin{array}{llllll}
   \dot{x_0}=x_0y_0 & & & & & \dot{y_0}=2x_0-2x_1 \\
   \dot{x_1}=x_1y_1 & & & & & \dot{y_1}=-x_0+2x_1-2x_2 \\
   \dot{x_2}=x_2y_2 & & & & & \dot{y_2}=-x_1 + 2x_2
\end{array}
\end{equation}
We denote by \(\mathcal{V}_1\) the vector field defined by the above
differential equations \eqref{systa4}. Then \(\mathcal{V}_1\) is the
Hamiltonian vector field, with Hamiltonian function    \( F_2=y_0^2
+ 4y_2^2-4x_0-8x_1-16x_2 \)   \\ with respect to the Poisson
structure \(\{\cdot,\cdot\}\) defined by the following
skew-symmetric matrix
\begin{equation}\label{matrij2}
     J=\frac{1}{8}\left(\begin{array}{cccccc}
                       0 & 0 & 0 & 4x_0 & -2x_0 & 0 \\
                          0 & 0 & 0 & -2x_1 & 2x_1 & -x_1 \\
                          0 & 0 & 0 & 0 & -x_2 & x_2 \\
                          -4x_0 & 2x_1 & 0 & 0 & 0 & 0 \\
                          2x_0 & -2x_1 & x_2 & 0 & 0 & 0 \\
                          0 & x_1 & -x_2 & 0 & 0 & 0
                        \end{array}
     \right)
\end{equation}
This Poisson structure is given on \(\mathbb{C}^6\); the function
\(F_0=y_0 + 2y_1 + 2y_2\) is a Casimir, so that the hyperplane
\(\mathcal{H}\) is a Poisson subvariety. The rank of this Poisson
structure \(\{\cdot,\cdot\}\) is \(0\) on the three-dimensional
subspace \(\{x_0 = x_1 = x_2 = 0\}\); the rank is \(2\) on the three
four-dimensional subspaces: \(\{x_0 = x_1 = 0\}\), \(\{x_0 = x_2 =
0\}\) and \(\{x_1 = x_2 = 0\}\). Thus, for all points of
\(\mathcal{H}\) except the four subspaces above the rank is \(4\).
The vector field \(\mathcal{V}_1\) admits also the following two
constants of motion:
\begin{equation}\label{invar1a4}
\begin{array}{l}
   F_1=x_0x_1^2x_2^2 \\
   F_2=y_0^2 + 4y_2^2-4x_0-8x_1-16x_2\\
   F_3= \left(y_0^2-4x_0\right)\left(y_2^2-4x_2\right)-4x_1\left(y_0y_2-4x_2-x_1\right)
\end{array}
\end{equation}
\(F_1\) is a Casimir for \(\{\cdot,\cdot\}\), and the function
\(F_3\) generates a second Hamiltonian vector field
\(\mathcal{V}_2\), which commutes with \(\mathcal{V}_1\), given by
the differential equations
\begin{equation}
   \begin{array}{l}\label{syst3}
x_0^{'}=x_0y_2\left(y_0y_2-2x_1\right)-4x_0x_2y_0\\
x_1^{'}=-x_1y_1y_2\left(y_1+y_2\right) - x_1^2y_1+x_1\left(x_0y_2+2x_2y_0\right)\\
x_2^{'}=x_2\left(y_1+y_2\right)\left(\left(y_1+y_2\right)y_2+x_1\right)+x_0x_2y_0 \\
y_0^{'}=2\left(2x_1x_2+x_0y_2^2\right)+ x_1\left(2x_1-y_0y_2\right)-8x_0x_2\\
y_1^{'}=-x_0y_2^2+2x_2 \left( 3x_0-x_1 \right) +y_0y_2 \left(
x_1+x_2 \right)-2x_1^2+x_2y_0y_1\\
y_2^{'}=x_1y_2\left( y_1+y_2\right)+x_1^2-x_2\left( y_1+y_2\right)
-2x_2x_0
  \end{array}
\end{equation}

Hence the system \eqref{systa4}  is completely integrable in the
Livouille sense. It can be written as a Hamiltonian vector fields
\[\dot{z}=J\frac{\partial H}{\partial z}, z=\left(z_1,\cdots,z_6\right)^{\top}=\left( x_0
, x_1 , x_2 , y_0 , y_1 , y_2 \right)^{\top}\] where \(H=F_2.\) the
Hamiltonian structure is defined by the following Poisson bracket
\[ \{F,H\}=\left\langle\frac{\partial F}{\partial z}, J\frac{\partial H}
{\partial z}\right\rangle=\displaystyle\sum_{i,k=1}^6
J_{ik}\frac{\partial F}{\partial z_i}\frac{\partial H}{\partial z_k}
\] where $\frac{\partial H}{\partial z}=\left(\frac{\partial
H}{\partial x_0},\frac{\partial H}{\partial x_1},\frac{\partial
H}{\partial x_2},\frac{\partial H}{\partial y_0},\frac{\partial
H}{\partial y_1},\frac{\partial H}{\partial y_2}\right)^\top$ and
$J$ is an
antisymmetric matrix. \\

The vector field $\mathcal{V}_2$ admits the same constants of
      motion \eqref{invar1a4} and is in involution with $\mathcal{V}_1$ therefore $\{F_2,F_3\}=0$. The involution $\sigma$ defined on $\mathbb{C}^6$ by
      \[\sigma\left(x_0,x_1,x_2,y_0,y_1,y_2\right)=\left(x_0,x_1,x_2 ,
      -y_0,-y_1,-y_2\right)\label{involinvar}
      \]
 preserves the constants of motion $F_1, F_2$ and $F_3$, hence leave the fibers of the momentum map $F$ invariant. This involution can be restricts to the hyperplane $\mathcal{H}$.

 \begin{lem}\cite{lietap}
The system of differential equation \eqref{systa4} of the vector
field  $\mathcal{V}_1$ has three distincts families of homogeneous
Laurent solutions with weights depending on four $\left(dim
\mathcal{H}-1\right)$ free parameters.
\end{lem}
The set of regular values of the momentum map $\mathbf{F}$ is the
Zariski open subset $\Omega$ defined by
\[\begin{array}{cl}
    \Omega =& \left\{c=\left(c_1,c_2,c_3\right) \in \mathbb{C}^3 \mid c_1\neq 0 \mbox{ and }\right. \\
    & \left.256\left(3200000c_1^2+2000c_3 ^2c_2c_1-225c_3c_2^3c_1+c_3^5\right)+1728c_2^5c_1-32c_3^4c_2^2+c_3^3c_2^4\neq 0\right\}.
\end{array}\]
At a generic point $c=\left(c_1,c_2,c_3\right) \in \mathbb{C}^3$,
the fiber on $c\in \Omega$ of $\mathbf{F }$ is therefore:
\[\mathbb{F}_c:= \mathbb{F}^{-1}\left(c\right)=\bigcap_{i=1}^3
\{m\in \mathcal{H}:F_i(m)=c_i\} \]

Hence we have the following result which prove that Toda lattice
$a_4^{\left(2\right)}$ is a completely integrable system in the
Liouville sense.
\begin{prop}\cite{lietap}
For $c\in \Omega$, the fiber $\mathbb{F}_c$ over $c$ of the momentum
$F$ is a smooth affine variety of dimension $2$ and the rank of the
Poisson structure \eqref{matrij2} is maximal and equal to $4$ at
each point of $\mathbb{F}_c$ ; moreover the vector fields
$\mathcal{V}_1$ and $\mathcal{V}_2$  are independent at each point
of the fiber $\mathbb{F}_c$.
\end{prop}

\begin{prop}\cite{lietap}
$\left(\mathcal{H}, \{\cdot,\cdot\},\mathbf{F}\right)$ is a
completely integrable system
 describing the Toda lattice $a_4^{\left(2\right)}$ where
$\mathbf{F}=\left(F_1,F_2,F_3\right)$ and $\{\cdot,\cdot\}$ are
given respectively by \eqref{invar1a4} and \eqref{matrij2} with
commuting vector fields \eqref{systa4} and \eqref{syst3}.
\end{prop}

The algebraic complete integrability of the $a_4^{\left(2\right)}$
Toda lattice was established in \cite{lietap} by the following
theorem
\begin{thm}\cite{lietap}
Let $\left(\mathcal{H}, \{\cdot,\cdot\},\mathbf{F}\right)$ be an
integrable system
 describing the Toda lattice $a_4^{\left(2\right)}$ where
$\mathbf{F}=\left(F_1,F_2,F_3\right)$ and $\{\cdot,\cdot\}$ are
given respectively by \eqref{invar1a4} and \eqref{matrij2} with
commuting vector fields \eqref{systa4}.
\begin{itemize}
     \item [i)] $\left(\mathcal{H}, \{\cdot,\cdot\},\mathbf{F}\right)$ is a weight homogeneous
     algebraical completely integrable  system.
     \item [ii)] For $c\in \Omega$, the fiber $\mathbb{F}_c$ of its
     momentum map is completed in an abelian surface
     $\mathbb{T}^2_c$ (the Jacobian of the hyperelliptic curve (of genus two) $\overline{\Gamma_c}^{\left(2\right)} $)
     by the addition of a singular divisor
     $\mathcal{D}_c$ composed of three irreducible components: $\mathcal{D}_c^{\left(0\right)}$ defined by:
     \begin{equation*}
\begin{array}{ll}
    \Gamma_c^{\left(0\right)}:& 16d^2a^8-\left(256d^3+8d^2c_2\right)a^6+
\left(1536d^2+96dc_2+8c_3+c_2^2\right)d^2a^4
-\left(\left(8\left(8c_3+48dc_2+c_2^2\right.\right.\right.\\
   &\left.\left.\left.+512d^2\right)d+2c_2c_3\right)d^2+64c_1\right)a^2+\left(8d\left(c_2c_3+
   16dc_3+64d^2c_2 +512d^3+2dc_2^2\right)+c_3^2\right)d^2=0
\end{array}
\end{equation*}
      and
  $\mathcal{D}_c^{\left(1\right)} $ defined by:
\begin{equation*}
    \Gamma_c^{\left(1\right)}:256ad^3-\left(\left(4a^2-c_2\right)^2-16c_3\right)d^2-+64c_1=0,
\end{equation*}
two singular curves of respective genus $3$ and $4$ and one
     smooth curve and $\mathcal{D}_c^{\left(2\right)}$ defined by
\begin{equation*}
    \Gamma_c^{\left(2\right)}:e^4a^4-\left(8c_1+c_2e^2\right)a^2e^2-64e^5+4e^2c_1c_2+4c_3e^4+16c_1^2=0.
\end{equation*}
 of genus $2$ and isomorphic to $\overline{\Gamma_c^{\left(2\right)}} $. The curves intercept each other as indicated in
 figure:\\
\begin{center}

\includegraphics[width=5 cm]{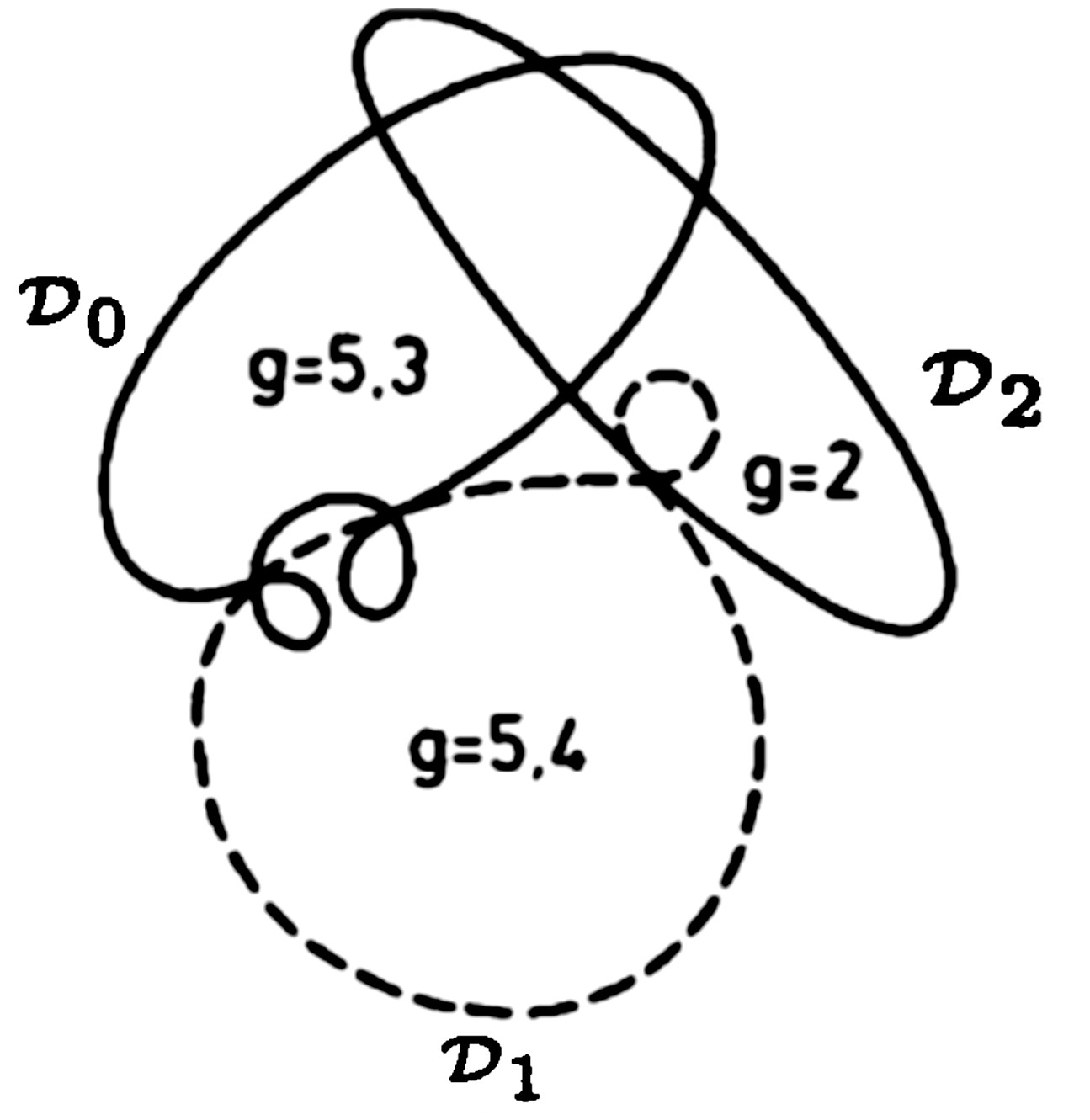}\\
  Figure: Curves completing the invariant surfaces \(\mathbb{F}_c\)
of the Toda lattice \(a_4^{\left(2\right)}\) in abelian surfaces
where \(\mathcal{D}_i\) is the curve
\(\mathcal{D}_c^{\left(i\right)}\) .
\label{todacomplete}
\end{center}

\end{itemize}
\end{thm}

\section{Linearization and Lax pairs of the $a_4^{\left(2\right)}$
Toda lattice}

The involution $(-1)$ on the abelian surface give a singular
surface, his Kummer surface. Here we give an equation of Kummer
surface lie with the Jacobi surface
$\mathcal{T}_c^2=Jac(\overline{\Gamma}_c)$ where
$\overline{\Gamma}_c$ is a hyperelliptic Riemann surface of genus
$2$ define above. The surface $\mathcal{T}_c^2$ is an abelian
principal polarisation and the section of the line bundle
$\left[2\mathcal{D}_c^{(2)}\right]$ embed his Kummer surface in the
projectif space $\mathcal{P}^6$
\par Consider the functions which have a double pole on one of
component of divisor $\mathcal{D}_c$, namely $\mathcal{D}_c^{(2)}$
and no pole on the other.\\

Now, we find a basis function on $\mathcal{H}$ which has a double
pole in $t$ when we substitute the principal balance $x(t;m_2)$ and
no poles when the other principal balances are substituted. Using
$x(t,m_0)$, $x(t,m_1)$ and $x(t,m_2)$ give in \cite{lietap}, we
obtain a basis of these functions constituate by the functions
$\theta_i$ give in the following table:
$$\begin{tabular}{|c|c|c|c|c|c|c|}
                             \hline
                             k       & $\dim \mathcal{F}^k$ & $\dim \mathcal{H}^k$  & $\dim \mathcal{Z}^k_\rho$  & $\sharp dep$ & $\zeta^k$ & \mbox{ indep. functions} \\
                             \hline
                             0       &       1              &         1             &             1              &      0       &       1   & $\theta_0$ \\
                             \hline
                             1       &       2              &         0             &             0              &      0       &       0   & - \\
                             \hline
                             2       &       6              &         1             &             2              &      1       &       1   & $\theta_1$ \\
                             \hline
                             3       &      10              &         0             &             0              &      0       &       0   & - \\
                             \hline
                             4       &      20              &         2             &             4              &      3       &       1   & $\theta_2$\\
                             \hline
                             5       &      30              &         0             &             0              &      0       &       0   & - \\
                             \hline
                             6       &      50              &         2             &            6              &      5       &        1   & $\theta_3$\\
                             \hline
                             7       &      70              &         0             &            0               &      0       &       0   & - \\
                             \hline
                             8       &     105              &         3             &            8              &     8       &         0   & -\\
                             \hline
                                                       \end{tabular}
$$
\begin{equation}\label{fonctionstheta}
    \begin{array}{cl}
       \theta_0 =& 1 \\
       \theta_1 =& x_2  \\
       \theta_2 =& x_1x_2+4x_2^2-y_2^2x_2 \\
       \theta_3 =& x_1x_2^2
    \end{array}
\end{equation}
The four functions $\theta_i$ are the line bundle section
$\left[2\mathcal{D}_c^{(2)}\right]$.\\
Hence we can formulate the following result:
\begin{prop}
The Koidara map which correspond to these functions:
$$\begin{array}{cccc}
  \psi_c: & Jac(\overline{\Gamma}_c) & \longrightarrow & \mathcal{P}^3\\
          & m=(x_0,x_1,x_2,y_0,y_2)           & \longmapsto     &
          (\theta_0(m):\theta_1(m):\theta_2(m):\theta_3(m)),
\end{array}$$
applied the Jacobi surface
$\mathcal{T}_c^2=Jac(\overline{\Gamma}_c)$ on his Kummer surface,
which is a singular quartic in the projective space  $
\mathcal{P}^3$. The basis $(\theta_0:\theta_1:\theta_2:\theta_3)$ is
taking convenably.
\end{prop}

\begin{proof}
 By substitute the balance $x(t,m_2)$ in the $ \theta_i$, $i=0,...,3$ functions
and taking the coefficients of $t^{-2}$ of Laurent series
$\theta_i(t,m_2)$, the map $\psi_c$ induce on $\Gamma_c$ a map
$$\begin{array}{cccc}
  \psi_c^{(2)}:  & (a,e)           & \longmapsto     &
          \left(0:1:\frac{1}{4e^2}\left(a^2e^2-c_2e^2-4c_1\right):e\right)
\end{array}.$$

Consider a Weierstrass point on $\overline{\Gamma}_c$  $\infty :
a=\varsigma^{-1}\mbox{  } \mbox{ , }   e=
\frac{1}{64}\left(\varsigma^{-4}-c_2\varsigma^{-2}+4c_3+O\left(\varsigma^6\right)\right)$.
 we obtain
 \begin{eqnarray}
 \begin{array}{ccl}
   \psi_c^{(2)}(\infty)&=&\lim \limits_{\varsigma\rightarrow
  0}\left(0:64\varsigma^4:16\varsigma^2-16c_2\varsigma^4+ O(\varsigma^6):1-c_2\varsigma^2+O(\varsigma^4)\right) \\
                       &=&\left(0:0:0:1\right)
 \end{array}
 \end{eqnarray}
 hence a basis $(\theta_0:\theta_1:\theta_2:\theta_3)$ is take convenably.
 \end{proof}
 Consider the constants of motion
 \begin{equation}\label{constantemouvement}
\begin{array}{cl}
  F_1 =& x_0x_1^2x_2^2=c_1 \\
  F_2 =& y_0^2 + 4y_2^2-4x_0-8x_1-16x_2=c_2 \\
  F_3 =& (y_0^2-4x_0)(y_2^2-4x_2)-4x_1(y_0y_2-4x_2-x_1)=c_3
\end{array}
 \end{equation}
 and eliminating the variables $(x_0,x_1,x_2,y_0,y_2)$ in the
principals balances $x(t,m_0)$, $x(t,m_1)$ and $x(t,m_2)$ in
\cite{lietap} we obtain:
\begin{equation}\label{variable_x_enfonctiontheta}
    x_0=\frac{c_1\theta_1^2}{\theta_3^2}\mbox{     }\mbox{   ,  }\mbox{     }
    x_1=\frac{\theta_3}{\theta_1^2}\mbox{     } \mbox{   ,  }\mbox{     }
    x_2= \theta_1
\end{equation}
Using the second equations of (\ref{fonctionstheta}) and
(\ref{constantemouvement}), we obtain
\begin{equation}\label{variable_y_enfonctiontheta}
   y_0^2= \frac{1}{\theta_1^2\theta_3^2}\left(4c_1\theta_1^4+4\theta_3^3+\theta_1\theta_3^2\left(c_2\theta_1+4\theta_2\right)\right)
  \mbox{     }\mbox{     }\mbox{     }\mbox{  ,   }\mbox{      }\mbox{     } y_2^2=\frac{1}{\theta_1^2}\left(4\theta_1^3-\theta_1\theta_2+\theta_3\right)
    \end{equation}
 Rewriting the last equation of
    (\ref{constantemouvement}) on the follow form
    $$4x_1y_0y_2=\left((y_0^2-4x_0)(y_2^2-4x_2)-c_3\right)+4x_1\left(4x_2+x_1\right)$$
we obtain a Kummer surface of Jac($\overline{\Gamma}_c$). It can be
put in the follow form
\begin{equation}\label{surfkummer}
\left(\left(c_2+16\theta_1\right)^2-16\left(16\theta_2+4c_2\theta_1+c_3\right)\right)\theta_3^2
+2\theta_3f_3\left(\theta_1,\theta_2\right)+f_4\left(\theta_1,\theta_2\right)=0
\end{equation}
where $f_3$ is a polynomial of degre $3$, $f_4$ of degre $4$ in
$\theta_1$ and $\theta_2$ given by
\begin{eqnarray*}
 f_3\left(\theta_1,\theta_2\right) &=& -\left(c_2+16\theta_1\right)\left(\theta_2\left(\theta_1c_2+4\theta_2\right)+c_3\theta_1^2\right)-64c_1 \\
 f_4\left(\theta_1,\theta_2\right) &=&
\left(c_3\theta_1^2+4\theta_2^2\right)^2-\theta_1\left(-2\theta_1^2c_2\theta_2c_3+256c_1\theta_1^2-\theta_1c_2^2\theta_2^2-64c_1\theta_2-8\theta_2^3c_2\right)
\end{eqnarray*}
Hence we have the following results:

\begin{prop}
A quartic equation of the Kummer surface of
Jac($\overline{\Gamma}_c$), in terms of $\theta_i$ is given by
$$\left(\left(c_2+16\theta_1\right)^2-16\left(16\theta_2+4c_2\theta_1+c_3\right)\right)\theta_3^2
+2\theta_3f_3\left(\theta_1,\theta_2\right)+f_4\left(\theta_1,\theta_2\right)=0$$
where $f_3$ is a polynomial of degre $3$, $f_4$ of degre $4$ in
$\theta_1$ and $\theta_2$ given by
\begin{eqnarray*}
 f_3\left(\theta_1,\theta_2\right) &=& -\left(c_2+16\theta_1\right)\left(\theta_2\left(\theta_1c_2+4\theta_2\right)+c_3\theta_1^2\right)-64c_1 \\
 f_4\left(\theta_1,\theta_2\right) &=&
\left(c_3\theta_1^2+4\theta_2^2\right)^2-\theta_1\left(-2\theta_1^2c_2\theta_2c_3+256c_1\theta_1^2-\theta_1c_2^2\theta_2^2-64c_1\theta_2-8\theta_2^3c_2\right)
\end{eqnarray*}
\end{prop}

\begin{thm}
The vector field $\mathcal{V}_1$ \ref{systa4} extends to a linear
vector field on the abelian surface $\mathbb{T}^2_c$ and the Jacobi
form for the differentials equation can be written as
$$\left\{\begin{array}{l}
  \frac{\dot{\lambda_1}}{\sqrt{f(\lambda_1)}}+\frac{\dot{\lambda_2}}{\sqrt{f(\lambda_2)}}= 0\\
  \frac{\lambda_1\dot{\lambda_1}}{\sqrt{f(\lambda_1)}}+\frac{\lambda_2\dot{\lambda_2}}{\sqrt{f(\lambda_2)}}= \frac{1}{2i}dt
\end{array}\right.$$
with
$f(\lambda)=\lambda_i^5+2c_2\lambda_i^4+\left(8c_3+c_2^2\right)\lambda_i^3+8c_2c_3\lambda_i^2+16c_3^2\lambda_i-16384c_1$
and $v^2=f(\lambda)$  is birational equivalent to the hyperelliptic
curve of genus two $\mathcal{K}_c$
\end{thm}
\begin{proof}
 Consider coefficient of $\theta_3^2$ in equation (\ref{surfkummer}) with the variables
$x_i$ and $y_i$
$$\Delta=\left(c_2+16x_2\right)^2-4\left(4x_2\left(-16y_2^2+64x_2+16x_1+4c_2\right)+4c_3\right)$$

 Let $u(\lambda)$ an unitary polynomial in $\lambda$ such that the
discriminant is $\Delta$, hence we have:
\begin{eqnarray*}
  u(\lambda) &=& \lambda^2+\left(c_2+16x_2\right)\lambda+
  4x_2\left(-16y_2^2+64x_2+16x_1+4c_2\right)+4c_3\\
             &=& \lambda^2+ \left(y_0^2 + 4y_2^2-4x_0-8x_1\right)\lambda+
  \left(4x_1-2y_0y_2\right)^2-16x_0y_2^2
\end{eqnarray*}.
\par Lets $\lambda_1$ and $\lambda_2$ roots of polynomial $f(\lambda)$, we have:
\begin{equation}\label{sommeetproduitulamda}
    \lambda_1+\lambda_2=-16x_2-c_2 \mbox{     }\mbox{     }\mbox{     }\mbox{     }\mbox{  ,   }
    \mbox{     }\mbox{     }\mbox{     } \lambda_1\lambda_2= 4x_2\left(-16y_2^2+64x_2+16x_1+4c_2\right)+4c_3
\end{equation}
 that imply,with respect with
 $\mathcal{V}_1$
 \begin{equation}\label{derivationsommeproduit}
    \dot{\lambda}_1+\dot{\lambda}_2=-16x_2y_2 \mbox{     }\mbox{     }\mbox{     }\mbox{     }\mbox{  ,   }
    \mbox{     }\mbox{     }\mbox{     } \dot{\lambda}_1\lambda_2+\lambda_1\dot{\lambda}_2=
    -16x_2\left(-y_2\left(y_0^2-4x_0\right)+2x_1y_0\right)
 \end{equation}
 Let $v\left(\lambda\right)$ a polynom define,up to a multiplicative
 constante, by :
 \begin{eqnarray}
 \begin{array}{ccc}
  v\left(\lambda\right)&=& 32i\left[x_2y_2\lambda+x_2\left(y_2\left(y_0^2-4x_0\right)-2x_1y_0\right)\right] \\
                       &=&
                       -2i\left(\dot{\lambda}_1+\dot{\lambda}_2\right)\lambda+2i\left(\dot{\lambda}_1\lambda_2+\lambda_1\dot{\lambda}_2\right).
 \end{array}
 \end{eqnarray}
by substituation (\ref{sommeetproduitulamda}) and
(\ref{derivationsommeproduit}) in (\ref{constantemouvement}), and by
eliminating variables $x_0,x_1,x_2,y_0$ and $y_2$, we obtain two
quadrics polynoms in $\dot{\lambda}_i^2$ given by
$$\dot{ \lambda }_i^2=\frac{\lambda_i^5+2c_2\lambda_i^4+\left(8c_3+c_2^2\right)\lambda_i^3+8c_2c_3\lambda_i^2+16c_3^2\lambda_i-16384c_1}{4\left(\lambda_1-\lambda_2\right)^2} \mbox{  ,   } i=1,2$$
verify
\begin{equation}\label{systmumford}
\left\{\begin{array}{l}
  \frac{\dot{\lambda_1}}{\sqrt{f(\lambda_1)}}+\frac{\dot{\lambda_2}}{\sqrt{f(\lambda_2)}}= 0\\
  \frac{\lambda_1\dot{\lambda_1}}{\sqrt{f(\lambda_1)}}+\frac{\lambda_2\dot{\lambda_2}}{\sqrt{f(\lambda_2)}}= \frac{1}{2i}dt
\end{array}\right.
\end{equation}
with
$$f(\lambda)=\lambda_i^5+2c_2\lambda_i^4+\left(8c_3+c_2^2\right)\lambda_i^3+8c_2c_3\lambda_i^2+16c_3^2\lambda_i-16384c_1$$
and like $v^2=f(\lambda)$ then:
\begin{eqnarray*}
  \sqrt{f\left(\lambda_l\right)} &=& v\left(\lambda_l\right) \\
                           &=& 2i\left[16x_2y_2\lambda_l+16ix_2\left(y_2\left(y_0^2-4x_0\right)-2x_1y_0\right)\right] \\
                           &=& -2i\left(\dot{\lambda}_1+\dot{\lambda}_2\right)\lambda_l+2i\left(\dot{\lambda}_1\lambda_2+\lambda_1\dot{\lambda}_2\right)
\end{eqnarray*}

hence
$$ \left\{
  \begin{array}{l}
    \sqrt{f\left(\lambda_1\right)}=-2i\left(\dot{\lambda}_1+\dot{\lambda}_2\right)\lambda_1+2i\left(\dot{\lambda}_1\lambda_2+\lambda_1\dot{\lambda}_2\right) \\
    \sqrt{f\left(\lambda_2\right)}=-2i\left(\dot{\lambda}_1+\dot{\lambda}_2\right)\lambda_2+2i\left(\dot{\lambda}_1\lambda_2+\lambda_1\dot{\lambda}_2\right)
  \end{array}\right.
  \Longrightarrow
  \left\{
  \begin{array}{l}
    \sqrt{f\left(\lambda_1\right)}=-2i\left(\lambda_1-\lambda_2\right)\dot{\lambda}_1\\
    \sqrt{f\left(\lambda_2\right)}=2i\left(\lambda_1-\lambda_2\right)\dot{\lambda}_2
  \end{array}
\right.
$$

 This show that the  Toda is linearising on the Jacobian variety of
 the curve $\mathcal{K}_c $. It is able to see how $\mathcal{K}_c $ and $v^2=f(s)$ are related.

 Like $$ \mathcal{K}_c: z^2=h(t)= t^5-2c_2t^4+\left(8c_3+c_2^2\right)t^3-8c_2c_3t^2+16c_3^2t+16384c_1
$$ and $$v^2=f(\lambda)= \lambda^5+2c_2\lambda^4+\left(8c_3+c_2^2\right)\lambda^3+8c_2c_3\lambda^2+16c_3^2\lambda-16384c_1
$$
then we easy verify by taking $\lambda=-t$ that $z=iv$ .\\

one verifies, by a direct computation, that the expression
$f(\lambda)-v^2(\lambda)$ is divisible by $u(\lambda)$ with
$$f(\lambda)= \lambda^5+2c_2\lambda^4+\left(8c_3+c_2^2\right)\lambda^3+8c_2c_3\lambda^2+16c_3^2\lambda-16384c_1
$$ Hence $y^2=f(\lambda)$ is birational to the affine curve
$\Gamma_c$ by adding the Weierstrass points at infinity
$a=\pm\sqrt{\frac{t^5-2c_2t^4+\left(8c_3+c_2^2\right)t^3-8c_2c_3t^2+16c_3^2t+16384c_1}{\left(t^2
-c_2t+4c_3\right)^2}}, \mbox{ } \mbox{ }\mbox{
}e=\frac{1}{64}\left(t^2-c_2t+4c_3\right)$.
\end{proof}

The form \ref{systmumford} is ewuivalent to
$$\frac{d}{dt}\left(\sum^2_{k=1}\int_{0_k}^{Q_k}\overrightarrow{\omega} \right)=\left(
                                                                                                  \begin{array}{c}
                                                                                                 0 \\
                                                                                                    2i \\
                                                                                                  \end{array}
                                                                                                    \right)
$$
where $\overrightarrow{\omega}=\left(
\frac{dx}{\sqrt{f(x)}},\frac{xdx}{\sqrt{f(x)}}\right)^{\top}$
 is a basis for holomorphic differentials on $\overline{\Gamma}_c$,
 $Q_1:=\left(\lambda_1, \sqrt{f(\lambda_1)}\right)$ and $Q_2:=\left(\lambda_2,
 \sqrt{f(\lambda_2)}\right)$ two points of $\Gamma_c$ and $Q_1+Q_2=\left(\lambda_1,
 \sqrt{f(\lambda_1)}\right)+\left(\lambda_2,
 \sqrt{f(\lambda_2)}\right)$  viewed as a
divisor on the genus $2$ hyperelliptic curve $\Gamma_c$. Thus, by
integrating \ref{systmumford}, we see that the flow of
$\overline{\mathcal{V}}_1$ is linear on the Jacobian of the curve
$\Gamma_c$. By using [\cite{mumford}, Theorem 5.3], one shows that
the symmetric functions $\lambda_1$ and $\lambda_2$, and hence the
original phase variables can be written in terms of theta functions.

 Now we also establish a link between the
$a_4^{\left(2\right)}$ Toda lattice and the Mumford system
\cite{mumford}. By using a method due to Vanhaecke \cite{Vanhaecke},
we construct an explicit morphism between these two systems. Thus,
we obtain a new Poisson structure for the Mumford system and then
derive a new Lax equation for the $a_4^{\left(2\right)}$ Toda
lattice.

According the fact that the expression $f(\lambda)-v^2(\lambda)$ is
divisible by $u(\lambda)$ such that the above formulas define a
point of $Jac(\overline{\Gamma}^2_c)\backslash \Gamma^2_c$, there
exist a polynomial $w$ in $\lambda$ of degree $3=$deg $u+1$. By
direct calculation, we obtain:
\begin{eqnarray*}
   w(\lambda) &=& \frac{ f(\lambda)- v^2(\lambda)}{ u(\lambda)} \\
              &=& \lambda^3+w_2\lambda^2+w_1\lambda+w_0,
\end{eqnarray*}
where
\begin{eqnarray*}
   w_0  &= & 256y_0^2x_2^2-1024x_0x_2^2 \\
    w_1 &= & 16x_1^2+4y_0^2y_2^2-32y_0^2x_2-16x_0y_2^2+128x_0x_2+256x_2^2+128x_2x_1-16x_1y_0y_2\\
    w_2 &= & y_0^2-8x_1-32x_2-4x_0+4y_2^2.
\end{eqnarray*}
The linearizing variables \ref{sommeetproduitulamda} and
\ref{derivationsommeproduit} suggest a morphism $\phi$ from the
$a^{(2)}_4$ Toda lattice to genus $2$ odd Mumford system:

$$\left\{\left(
          \begin{array}{cc}
             v(\lambda) &   u(\lambda) \\
             w(\lambda) &  -v(\lambda) \\
          \end{array}
        \right)\in M_2(\mathbb{C}[\lambda]) \mbox{ such that }
         \begin{array}{c}
                 deg(u)=2=deg(w)-1 \\
                 deg(v)<2; \mbox{u,w are monic }
          \end{array}\right\}\cong \mathbb{C}^7,
$$
 where $\mathbb{C}^7$ is a phase space of Mumford system. The morphism $\phi$ is given
 by:
\begin{equation}\label{homomorphismemumford}
    \begin{array}{ccc}
      (x_0,x_1,x_2,y_0,y_2) & \longmapsto & \left\{
                                              \begin{array}{ll}
                                                u(\lambda)= & \lambda^2+u_1\lambda+u_0 \\
                                                v(\lambda)= & v_1\lambda+v_0  \\
                                                w(\lambda)= & \lambda^3+w_2\lambda^2+w_1\lambda+w_0
                                              \end{array}
                                            \right.
    \end{array}
\end{equation}
 with
 $$\begin{array}{llllll}
    u_0= & \left(4x_1-2y_0y_2\right)^2-16x_0y_2^2 &  &  & v_0= & 16x_2\left(y_2\left(y_0^2-4x_0\right)-2x_1y_0\right) \\
    u_1= & -\left(y_0^2 + 4y_2^2-4x_0-8x_1\right)               &  &  & v_1= & 16x_2y_2
   \end{array}
 $$  $$\begin{array}{ll}
    w_0= & 256y_0^2x_2^2-1024x_0x_2^2\\
    w_1= & 16x_1^2+4y_0^2y_2^2-32y_0^2x_2-16x_0y_2^2+128x_0x_2+256x_2^2+128x_2x_1-16x_1y_0y_2\\
    w_2= & y_0^2-8x_1-32x_2-4x_0+4y_2^2
   \end{array}
 $$
\begin{thm}
A Lax representation of the vector field
$\mathcal{V}_1=\mathcal{X}_{F_1}$ is given by:
$$\dot{X}=\left[X(\lambda),Y(\lambda)\right]$$
by taking $$X(\lambda)=\left(
          \begin{array}{cc}
             v(\lambda) &   u(\lambda) \\
             w(\lambda) &  -v(\lambda) \\
          \end{array}
        \right) \mbox{ and } Y(\lambda)=\left(
          \begin{array}{cc}
             0 &   1 \\
             b(\lambda) &  0 \\
          \end{array}
        \right)$$
\end{thm}
where $u(\lambda),v(\lambda)$ and $w(\lambda)$ are the polynomials
defined above. The coefficient $b(\lambda)=\lambda-32x_2$ of the
matrix $Y(\lambda)$ is the polynomial part of the rational function
$w(\lambda)/u(\lambda)$.


\subsection*{Acknowledgements}
We would like to extend our sincere gratitude to Professor Pol
Vanhaecke at University of Poitiers for his particular contributions
in providing clarifications and guidance on our research theme,  for
the enriching exchanges and thoughtful advice he generously offered
us throughout this project.

\end{document}